\documentclass{article}
\usepackage[utf8]{inputenc}
\usepackage{hyperref,amsthm,amsmath,amssymb,amsfonts}
\usepackage{aliascnt,cleveref}
\usepackage{fullpage}
\usepackage{color}
\usepackage{xcolor}
\usepackage{mathrsfs}

\newtheorem{theorem}{Theorem}[section]

\newtheorem{definition}[theorem]{Definition}

\newtheorem{remark}[theorem]{Remark}
\newtheorem{question}[theorem]{Question}

\newcommand{\AAA}{\mathcal{A}}

\newcommand{\BBB}{\mathcal B}
\newcommand{\eps}{\varepsilon}

\DeclareMathSymbol{\R}{\mathbin}{AMSb}{"52}

\newcommand{\Space}{{\rm Space}}
\newcommand{\polylog}{{\rm polylog}}
\newcommand{\thresh}{{\rm THRESH}}

\title{A Note on Sanitizing Streams with Differential Privacy}
\author{
Haim Kaplan
\and
Uri Stemmer
}
\date{November 26, 2021}

\begin{document}

\maketitle

\begin{abstract}
The literature on {\em data sanitization} aims to design algorithms that take an input dataset and produce a privacy-preserving version of it, that captures some of its statistical properties. In this note we study this question from a streaming perspective and our  goal is to {\em sanitize a data stream}. Specifically, we consider low-memory algorithms that operate on a data stream and  produce an alternative privacy-preserving stream that captures some statistical properties of the original input stream.

\end{abstract}

\section{Introduction and Notations}
Data sanitization is one of the most well-studied topics in the literature of differential privacy. Informally, the goal is to design differentially private algorithms that take a dataset $D$ and produce an alternative dataset $D'$ which captures some statistical properties of $D$. Most of the research on data sanitization focuses on the ``offline'' setting, where the algorithm has full access to the input dataset $S$. We continue the study of this question from a streaming perspective. Specifically, we aim to design low-memory algorithms that operate on a large stream of data points, and produce a privacy-preserving variant of the original stream while capturing some of its statistical properties. We present a simple reduction from the streaming setting to the offline setting. As a special case, we improve on the recent work of Alabi et al.~\cite{ABC21}, who studied a special case of this question.

Let us begin by recalling basic notations from the literature on data sanitization. Let $X$ be a data domain and let $C$ be a class of predicates, where each $c\in C$ maps $X$ to $\{0,1\}$. Given a dataset $D\in X^*$, a {\em sanitization mechanism} for $C$ is required to produce a {\em synthetic dataset} $D'$ such that for every $c\in C$ we have
$$
\frac{1}{|D|}\sum_{x\in D} c(x) \approx \frac{1}{|D'|}\sum_{x\in D'} c(x).
$$
Formally,

\begin{definition}
Let $C$ be a class of predicates mapping $X$ to $\{0,1\}$. Let $\AAA$ be an algorithm that takes an input dataset $D\in X^n$ and outputs a dataset $D'\in X^n$. Algorithm $\AAA$ is an $(\alpha,\beta,\eps,\delta,n)$-sanitizer for $C$, if
\begin{enumerate}
\item $\AAA$ is $(\eps,\delta)$-differentially private (see \cite{DworkMNS16,DworkR14,Vadhan17} for background on differential privacy).
\item For every input $D \in X^n$ and for every predicate $c\in C$
we have \\
$\Pr\limits_{\AAA}\left[
\left|\frac{1}{|D|}\sum_{x\in D} c(x) - \frac{1}{|D'|}\sum_{x\in D'} c(x)\right|\leq\alpha
\right]\geq 1-\beta$. The probability is over the coin tosses of $\AAA$.
\end{enumerate}
\end{definition}

\begin{remark}
It is often convenient to allow the the size of $D'$ to be different than $n$ (the size of the original dataset). For simplicity, here we assume that $\AAA$ produces a dataset $D'$ of the same size as $D$.
\end{remark}

We consider a variant of the above definition where the input of algorithm $\AAA$ is a stream, and its output is also a stream. The utility requirement is that at the end of the stream, the produced stream $S'$ is similar to the the original stream $S$ (in the same sense as above). Formally,

\begin{definition}
Let $C$ be a class of predicates mapping $X$ to $\{0,1\}$. Let $\AAA$ be an algorithm that operates on a stream $S\in X^m$ outputs a stream $S'\in X^m$. Algorithm $\AAA$ is an $(\alpha,\beta,\eps,\delta,m)$-streaming-sanitizer for $C$, if
\begin{enumerate}
\item $\AAA$ is $(\eps,\delta)$-differentially private;
\item For every input stream $S \in X^m$ of length $m$ and for every predicate $c\in C$
we have \\
$\Pr\limits_{\AAA}\left[
\left|\frac{1}{|S|}\sum_{x\in S} c(x) - \frac{1}{|S'|}\sum_{x\in S'} c(x)\right|\leq\alpha
\right]\geq 1-\beta$. The probability is over the coin tosses of $\AAA$.
\end{enumerate}
\end{definition}

\begin{remark}\label{rem:conti}
For simplicity, in the above definition we required utility to hold only at the end of the stream. Our results remain essentially unchanged also with a variant of the definition where utility must hold at any moment throughout the execution. We can also allow the size of the output stream to be different than the size of the input stream.
\end{remark}

\section{A Generic Reduction to the Offline Setting}

We observe that every (offline) sanitizer can be transformed into a streaming-sanitizer as follows.

\begin{theorem}\label{thm:main}
Let $\AAA$ be an $(\alpha,\beta,\eps,\delta,n)$-sanitizer for a class of predicates $C$, with space complexity $\Space(\AAA,n)$. Let  $m=k\cdot n$ for some positive integer $k$. Then, there exists an $(\alpha,k\beta,\eps,\delta,m)$-streaming-sanitizer for $C$ using space $\Space(\AAA,n)$ (we assume that $\AAA$'s space is at least linear).
\end{theorem}

\begin{remark}
The point here is that, even though the stream is large (of length $m$), the space complexity of the streaming-sanitizer essentially depends only on the  space complexity of the (offline) sanitizer when applied to ``small'' datasets of size $n$. (As Theorem~\ref{thm:main} is stated, the length of the stream $m$ affects the confidence parameter of the resulting streaming-sanitizer; however, in Remark~\ref{rem:rem} we explain how this can be avoided.)
\end{remark}

\begin{proof}[Proof of Theorem~\ref{thm:main}]
We construct a streaming sanitizer $\BBB$ as follows:
\begin{enumerate}
    \item Let $D\in X^n$ denote the next $n$ items in the stream.
    \item Output $D'=\AAA(D)$ and goto Step 1.
\end{enumerate}

First observe $\BBB$ applies $\AAA$ to disjoint portions of its input stream, and hence, algorithm $\BBB$ is $(\eps,\delta)$-differentially private. Next, by a union bound, with probability at least $1-k\beta$, we have that all of the applications of algorithm $\AAA$ succeed in producing a  dataset $D'$ that maintain averages of predicates in $C$ up to an error of $\alpha$. In such a case, for the entire input stream $S$ and output stream $S'$, and for every $c\in C$, it holds that
$$
\sum_{x\in S}c(x)
=
\sum_{i=1}^k \sum_{x\in D}c(x)
\leq
\sum_{i=1}^k \left( \sum_{x\in D'}c(x)+\alpha n \right)
=
\sum_{x\in S'}c(x)+\alpha kn,
$$
and hence
$$
\frac{1}{|S|}\sum_{x\in S}c(x)\leq \frac{1}{|S'|}\sum_{x\in S'}c(x)+\alpha.
$$
The other direction is symmetric.
\end{proof}

\begin{remark}\label{rem:rem}
Two remarks are in order. First, assuming that $m$ is big enough w.r.t.\ $n$, we can avoid blowing up the confidence parameter $\beta$ by $k$. The reason is that, by the Chernoff bound, w.h.p.\ we get that at least $(1-2\beta)$ fraction of the executions of $\AAA$ succeed, in which case the overall error would be at most $\alpha+2\beta$. Second, again assuming that $m$ is big enough, we could relax the privacy guarantees of $\AAA$ while keeping $\BBB$'s privacy guarantees unchanged. Specifically, a well known fact is that we can boost the privacy guarantees of a differentially private algorithm by applying it to a subsample of its input dataset. Hence, assuming that $\AAA$ is $(1,\delta)$-differentially private, we can execute the above algorithm $\BBB$ on a random subsample $\tilde{S}$ from the original input stream $S$ (by selecting each element of $S$ independently with probability $\eps$). Assuming that $m\gtrsim\frac{1}{\alpha^2\eps}$, then the subsampled stream $\tilde{S}$ is big enough such that the additional error introduced by this subsampling is at most $\alpha$.
\end{remark}

\section{Application to Bounded Space Differentially Private Quantiles}
Recently, Alabi et al.~\cite{ABC21} introduced the problem of differentially private quantile estimation with sublinear space complexity. Specifically, let $\alpha$ be an approximation parameter, and consider an input stream $S$ containing $m$ points from a domain $X$. The goal is to design a small-space differentially-private algorithm that, at the end of the stream, is capable of approximating all quantiles in the data up to error $\alpha$ and confidence $\beta$. Specifically, \cite{ABC21} designed a private variant of the streaming algorithm of Greenwald and Khanna \cite{GreenwaldK01}. They obtained an $(
\eps,0)$-differentially private algorithm with space complexity\footnote{We use $\tilde{O}(f)$ to hide $\polylog(f)$ factors.} $\tilde{O}\left(\frac{1}{\alpha\eps}\log(\frac{|X|}{\beta})\log(m)\right)$, which is great because it matches the non-private space dependency of $1/\alpha$. However, the following questions were left open (and stated explicitly as open questions by Alabi et al.~\cite{ABC21}).

\begin{question}
The algorithm of \cite{ABC21} was tailored to the non-private algorithm of Greenwald and Khanna \cite{GreenwaldK01}, which is known to be sub-optimal in the non-private setting. Can we devise a more general approach that would allow us to instantiate (and benefit from) the state-of-the-art non-private algorithms? In particular, can we avoid the dependency of the space complexity in $\log(m)$?
\end{question}

\begin{question}
The space complexity of the algorithm of \cite{ABC21} grows with $\log|X|$. Can this be avoided?
\end{question}

We observe that using our notion of streaming-sanitizers immediately resolves these two questions. To see this, let $X$ be a totally-ordered data domain. For a point $x\in X$, let $c_x:X\rightarrow\{0,1\}$ be a {\em threshold function} defined by $c_x(y)=1$ iff $x\leq y$. Let $\thresh_X=\{c_x : x\in X\}$ denote the class of all threshold functions over $X$. This class captures all quantiles of the original stream. Now, to design a differentially private quantile estimation algorithm (with sublinear space), all we need to do is to apply our generic construction for a streaming-sanitizer for $\thresh_X$ (instantiated with the state-of-the-art offline sanitizer for this class), and to run any non-private streaming algorithm for quantiles estimation on the outcome of the streaming-sanitizer. Using the state-of-the-art offline sanitizer for $\thresh_X$ from \cite{KaplanLMNS20} and the state-of-the-art non-private streaming algorithm of \cite{KarninLL16}, we get an $(\eps,\delta)$-differentially-private quantiles-estimation algorithm using space complexity $\tilde{O}\left(\frac{1}{\alpha\eps}\log(\frac{1}{\beta})\left(\log(\frac{1}{\delta})\log^*|X|\right)^{1.5}\right)$.

\begin{remark}
The above idea is general, and is not restricted to the algorithms of \cite{KaplanLMNS20} and \cite{KarninLL16}. In particular, we  could  have used an $(\eps,0)$-differentially private sanitizer, at the expense of having the space complexity grow with $\log|X|$ instead of $(\log^*|X|)^{1.5}$. See, e.g.,~\cite{BeimelNS16,BunNSV15,KaplanScSt21} for additional constructions of (offline) sanitizers for $\thresh_X$.
\end{remark}

\begin{remark}
As in Remark~\ref{rem:rem}, assuming that the stream length $m$ is big enough, the space complexity can be made independent of $\varepsilon$.
\end{remark}

Alabi et al.~\cite{ABC21} also asked if it is possible to obtain a differentially private streaming algorithm for quantiles estimation that
allows for continually monitoring how the quantiles evolve throughout the stream (rather then only at the end of the stream). Specifically, the algorithm of \cite{ABC21} works in the ``one-shot'' setting
where the quantiles are computed once after observing the entire stream. However, many relevant applications require real-time release
of statistics. Using our notion of a streaming-sanitizer, this comes essentially for free (see Remark~\ref{rem:conti}), because once the stream is private, any post-processing of it satisfies privacy.

\bibliographystyle{alpha}
\newcommand{\etalchar}[1]{$^{#1}$}

\end{document}